\documentclass[10pt,letterpaper,aps,prd,notitlepage,tightenlines,nofootinbib]{revtex4-1}
\usepackage[english]{babel}

\usepackage[left=1in,right=1in,top=1in,bottom=1in]{geometry}
\usepackage{microtype}
\usepackage[utf8]{inputenc}
\usepackage{amsmath}
\usepackage{amsfonts}
\usepackage{amssymb}
\usepackage{amsthm}
\usepackage{braket}
\usepackage{bbm}
\usepackage{hyperref}
\usepackage{graphicx}
\usepackage{booktabs}
\usepackage{capt-of}

\renewcommand{\vec}{\mathbf}

\newcommand{\transpose}{\mathrm{t}}
\newcommand{\sltr}{\mathfrak{sl}(2,\mathbb{R})}

\DeclareFontFamily{U}{mathx}{\hyphenchar\font45}
\DeclareFontShape{U}{mathx}{m}{n}{<-> mathx10}{}
\DeclareSymbolFont{mathx}{U}{mathx}{m}{n}
\DeclareMathAccent{\widebar}{0}{mathx}{"73}

\theoremstyle{plain}
\newtheorem{proposition}{Proposition}
\newtheorem*{theorem*}{Theorem}
\newtheorem*{remark}{Remark}

\newtheorem{appxprop}{Proposition}
\newtheorem{lemma}{Lemma}

\newtheorem{appxcoroll}{Corollary}

\begin{document}
\title{Wigner--Eckart theorem for the non-compact algebra \texorpdfstring{$\mathfrak{sl}(2,\mathbb{R})$}{sl(2,R)}}

\author{Giuseppe Sellaroli}\email{gsellaroli@uwaterloo.ca}
\affiliation{Department of Applied Mathematics, University of Waterloo\\Waterloo, Ontario, Canada}

\date{\today}

\begin{abstract}
The Wigner--Eckart theorem is a well known result for tensor operators of $\mathfrak{su}(2)$ and, more generally, any compact Lie algebra. In this paper the theorem will be generalized to the particular non-compact case of $\sltr$. In order to do so, recoupling theory between representations that are not necessarily unitary will be studied, namely between finite-dimensional and infinite-dimensional representations. As an application, the Wigner--Eckart theorem will be used to construct an analogue of the Jordan--Schwinger representation, previously known only for representations in the discrete class, which also covers the continuous class.
\end{abstract}

\maketitle

\section*{Introduction}

Representation theory of Lie groups and Lie algebras have many applications in physics, especially in quantum theory. In particular,  tensor operators have been useful for a long time in non-relativistic quantum mechanics \cite{messiah2}, and have more recently been introduced in loop quantum gravity \cite{girelli1}. A remarkable property of tensor operators for compact Lie algebras, encoded in what is known as \emph{Wigner--Eckart theorem} \cite{barut}, is that their matrix elements can be expressed as a product of a Clebsch--Gordan coefficient\footnote{The coefficients appearing in the decomposition of the tensor product of two irreducible representations as the direct sum of irreducible representations, also known as \emph{recoupling theory}.} and a factor that does not depend on the particular basis vectors.

Although compact Lie algebras are usually considered, one may want to investigate the non-compact case.
%In particular, the possibility of a non-compact gauge group (and hence algebra) has been investigated in loop quantum gravity \cite{freidel,rovelli}.
An analogue of the Wigner--Eckart theorem for non-compact algebras has already been considered \citep{locallycompact}, but only for infinite-dimensional tensor operators, i.e. with infinitely many components. Here finite-dimensional tensor operators will be considered. A general result for every non-compact Lie algebra is not available in this case; instead, the ``simple'' case of $\sltr$ will considered, with the hope that the techniques introduced can be used to study more general algebras. The theorem itself has a relatively simple proof; however, it relies on the knowledge of recoupling theory of finite-dimensional and infinite-dimensional representations. Finite-dimensional representations are, apart from the trivial one, always non-unitary, while infinite-dimensional one can be unitary or not. For this reason, before the theorem can be proved, these recouplings, which were previously unconsidered (to the best of the author's knowledge\footnote{With the exception of \cite{haruo}, which however only considers the discrete series representations and, even then, is lacking some results presented here.}), will be studied. Most of the paper will be dedicated to this task.

An application of the theorem in the $\sltr$ case will also be presented. An important result of $\mathfrak{su}(2)$ representation theory, especially useful in quantum field theory, is the Jordan--Schwinger representation, which consists of expressing the algebra generators in terms of two uncoupled \emph{quantum harmonic oscillator} operators. A similar result for $\sltr$ exists, but only for certain representations classes \cite{boson}. It will be shown that, making use of the Wigner--Eckart theorem, \emph{all} representation classes admit an analogous construction in terms of two tensor operators, which reduce to the usual Jordan--Schwinger representation where the latter is defined.

The paper is organized as follows: section 1 is a review of preliminary notions, i.e. $\sltr$ representation theory and tensor operators. The main results of the paper are presented in section 2: recoupling theory between finite and infinite-dimensional representations, Wigner--Eckart theorem and Jordan--Schwinger representation. Finally a table of notations used in the paper, some results needed in section 2 and a table of Clebsch--Gordan coefficients for the coupling of finite and infinite-dimensional representations are included as appendices.

\section{Preliminary notions}
This section contains a review of some notions that will be used in the main part of the paper. Firstly the representation theory of $\sltr$ will be recalled, then tensor operators will be defined for a generic Lie algebra and, in particular, for $\sltr$. The review of representation theory follows \cite{Bargmann1947,harishchandra,howe1992non}, albeit with different notations and conventions. References for tensor operators can be found in \cite{jeevanjee2011,barut}.
\subsection{Irreducible representations of \texorpdfstring{$\mathfrak{sl}(2,\mathbb{R})$}{sl(2,R)}}

The $3$-dimensional real Lie algebra $\mathfrak{sl}(2,\mathbb{R})$ is the algebra of traceless $2\times 2$ real matrices; it is isomorphic to the real algebras $\mathfrak{spin}(2,1)$, $\mathfrak{so}(2,1)$, $\mathfrak{sp}(2,\mathbb{R})$ and $\mathfrak{su}(1,1)$. The non-standard basis
\begin{equation}
X_0=\tfrac{1}{2}
\begin{pmatrix}
0 & 1\\
-1 & 0
\end{pmatrix},
\quad X_1=\tfrac{1}{2}
\begin{pmatrix}
1 & & 0 \\
0 & & -1
\end{pmatrix},
\quad X_2=\tfrac{1}{2}
\begin{pmatrix}
0 && 1\\
1 && 0
\end{pmatrix}
\end{equation}
will be used here, with commutation relations
\begin{equation}
[X_0,X_1]=-X_2,\quad [X_1,X_2]=X_0,\quad [X_2,X_0]=-X_1.
\end{equation}
The Casimir operator is given in this basis by
\begin{equation}
\label{eq:casimir}
Q=(X_0)^2 -(X_1)^2 -(X_2)^2.
\end{equation}
For consistency with the literature on the subject, the usual physicist convention of acting on complex representations with complexified generators will be used. Explicitly, the new generators are
\begin{equation}\label{eq:phys_generators}
J_0:=-iX_0,\quad J_\pm:=-iX_1\pm X_2
\end{equation}
with commutation relations
\begin{equation}
[J_0,J_\pm]=\pm J_\pm,\quad[J_+,J_-]=-2 J_0,
\end{equation}
and the Casimir is given by
\begin{equation}
Q = -J_0(J_0+1)+J_-J_+\equiv-J_0(J_0-1)+J_+J_-.
\end{equation}
They act on complex irreducible representations (not necessarily unitary) as
\begin{equation}
\label{eq:reps}
\begin{cases}
J_0\ket{j,m}=m\ket{j,m}\\
J_\pm\ket{j,m}=C_\pm(j,m)\ket{j,m\pm 1}\\
Q\ket{j,m}=-j(j+1)\ket{j,m},
\end{cases}
\end{equation}
where
\begin{equation}
C_\pm(j,m):=i\sqrt{j \mp m}\sqrt{j \pm m + 1}.
\end{equation}
The vectors $\ket{j,m}$ form an orthonormal basis for the vector space of the representation, with $j$ being a label for the representation and $m$ enumerating the vectors; their possible values depend on the representation class, which can be one of the following:
\begin{itemize}
\item {\bf Positive discrete series} $D^+_j$: infinite-dimensional lowest weight\footnote{A lowest (highest) weight representation is, in this context, one for which $m$ has a lower (upper) bound.} representations, with
\begin{equation*}
j\in\left\{-\tfrac{1}{2},0,\tfrac{1}{2},1,\dotsc\right\} \quad \mbox{and}\quad m\in\left\{ j+1,j+2,j+3,\dotsc\right\}.
\end{equation*}
\item {\bf Negative discrete series} $D^-_j$: infinite-dimensional highest weight representations, with
\begin{equation*}
j\in\left\{-\tfrac{1}{2},0,\tfrac{1}{2},1,\dotsc\right\} \quad \mbox{and}\quad m\in\left\{-j-1,-j-2,-j-3,\dotsc\right\}.
\end{equation*}
\item {\bf Continuous series} $C_j^\varepsilon$: infinite-dimensional representations of \emph{parity} $\varepsilon\in\left\{0,\tfrac{1}{2}\right\}$, with
\begin{equation*}
m\in\varepsilon+\mathbb{Z}\quad\mbox{and}\quad j\in\mathbb{C};
\end{equation*}
when $j$ is (half-)integer, there is the additional constraint
\begin{equation*}
j-\varepsilon\not\in\mathbb{Z}.
\end{equation*}
Moreover, the representations $C^\varepsilon_j$ and $ C^\varepsilon_{-j-1}$ are isomorphic.
\item {\bf Finite-dimensional series} $F_j$: isomorphic to the representations of $\mathfrak{su}(2)$, with
\begin{equation*}
j\in\left\{0,\tfrac{1}{2},1,\dotsc\right\}\quad\mbox{and}\quad m\in\left\{-j,-j+1,\dotsc,j-1,j\right\}.
\end{equation*}
They are the only finite-dimensional representations, with dimension $2j+1$.
\end{itemize}
Of these representations, the only unitary ones are the whole discrete (positive and negative) series, the continuous series with
\begin{equation}
\begin{cases}
j\in \left(-1,0\right)\quad\mbox{or}\quad j\in\left\{-\tfrac{1}{2}+is \left|\right. s\neq0 \right\}\quad & \mbox{if}\quad \varepsilon=0\\
j\in \left\{-\tfrac{1}{2}+is \left|\right. s\neq0 \right\}\quad & \mbox{if}\quad \varepsilon=\tfrac{1}{2}
\end{cases}
\end{equation}
and, among the finite-dimensional ones, only the \emph{trivial representation} $F_0$.
\begin{remark}\label{rem:rep_group}
These representations can be integrated to representations of the group $\mathrm{SL}(2,\mathbb{R})$. In this case, the only ones appearing in the \emph{Plancherel decomposition} are the discrete ones with $j\geq 0$ and the continuous ones with $j\in\left\{-\tfrac{1}{2}+is \left|\right. s\neq0 \right\}$ (see \cite{harishchandra} for details).
\end{remark}

\subsection{Tensor operators for \texorpdfstring{$\mathfrak{sl}(2,\mathbb{R})$}{sl(2,R)}}
Tensor operators are a particular class of operators that transform as vectors in a representation of a Lie algebra under the ``action'' of the algebra generators. Explicitly, let
\begin{equation}
\rho_i:\mathfrak{g}\rightarrow \mathfrak{gl}(V_i),\qquad i=1,2
\end{equation}
be two representations of a Lie algebra $\mathfrak{g}$. One can always associate to them a new representation
\begin{equation}
R:\mathfrak{g}\rightarrow \mathfrak{gl}(\mathrm{Lin}(V_1, V_2))
\end{equation}
defined by
\begin{equation}
\label{eq:rep_operator}
R(X)\,A=\rho_2(X)\,A -A\,\rho_1(X),\qquad \forall A\in \mathrm{Lin}(V_1, V_2),\quad \forall X\in \mathfrak{g}.
\end{equation}
A tensor operator $T$ is an intertwiner between some representation
\begin{equation}
\rho_0:\mathfrak{g}\rightarrow \mathfrak{gl}(V_0)
\end{equation}
and $R$, i.e. a linear map
\begin{equation}
T:V_0\rightarrow \mathrm{Lin}(V_1,V_2)
\end{equation}
such that
\begin{equation}
T\circ\rho_0(X) = R(X)\circ T,\qquad \forall X\in \mathfrak{g}.
\end{equation}
If $\rho_0$ is irreducible, $T$ is called an irreducible tensor operator. As usual with linear maps, the components of a linear operator in a given basis are defined by its evaluation on the basis vectors.

\begin{remark}[Extension to the group]\label{rem:group_extension}
Tensor operators can also be defined for Lie group representations, in a similar way. Whenever the algebra representations $\rho_0$, $\rho_1$ and $\rho_2$ are also group representations, the two definitions are equivalent \cite{barut}.
\end{remark}

In the specific case of $\mathfrak{sl}(2,\mathbb{R})$, one says a tensor operator $T^\gamma$ is of rank $\gamma$ if $\rho_0$ is the finite-dimensional representation $F_\gamma$. Its components in the standard basis (\ref{eq:reps}) are given by
\begin{equation}
T^\gamma_\mu:=T^\gamma \left(\ket{\gamma,\mu}\right)\in\mathrm{Lin}(V_1,V_2),\qquad \mu\in\left\{-\gamma,\dotsc,\gamma \right\}.
\end{equation}
Owing to (\ref{eq:rep_operator}), they satisfy for all $X\in\mathfrak{sl}(2,\mathbb{R})$
\begin{equation}
R(X)\,T^\gamma_\mu \equiv \rho_2(X)\,T^\gamma_\mu - T^\gamma_\mu\,\rho_1(X)=\sum_{\nu=-\gamma}^\gamma \braket{\gamma,\nu|F_\gamma(X)|\gamma,\mu}T^\gamma_\nu.
\end{equation}
In terms of the algebra generators, this can be written in the compact form
\begin{equation}\label{eq:tenso_op_def}
[J_0,T^\gamma_\mu]=\mu\,T^\gamma_\mu ,\qquad[J_\pm,T^\gamma_\mu]=C_\pm(\gamma,\mu)\,T^\gamma_{\mu\pm1}.
\end{equation}

\begin{remark}[Infinite-dimensional tensor operators]\label{rem:infinite_t_op}
Although out the scope of this paper, one could also consider tensor operators where $\rho_0$ is an infinite-dimensional representation. In particular, when the representation is unitary, it can be proven that, even for non-compact groups (and hence algebras) the Wigner--Eckart theorem holds \cite{locallycompact}.
\end{remark}

\section{Wigner--Eckart theorem}
This section contains the main result of the paper, that is the Wigner--Eckart theorem for $\sltr$. The theorem roughly states that the matrix elements of a tensor operators between two representations (of any class) are heavily constrained by the way the finite-dimensional representation which the operator transforms like couples with the one on which it is acting.

Coupling of two finite-dimensional representations are known, as they behave exactly like $\mathfrak{su}(2)$ representations. In order to prove the theorem, couplings of finite-dimensional representations and infinite-dimensional will be studied here.

The section is organized as follows: first the couplings $F_\gamma\otimes D^+_j$ and $F_\gamma\otimes C^\varepsilon_j$ will be discussed, then the Wigner--Eckart theorem will be properly stated and proved. Lastly, as an application, the theorem will be used to generalize the Jordan--Schwinger representation of $\mathfrak{su}(2)$ to the non-compact $\sltr$.

\subsection{Coupling of finite and discrete representations}\label{sec:fin-dis}
Consider the coupling $F_\gamma\otimes D^+_j$ of a finite-dimensional representation and one from the discrete positive series, with $\gamma\geq\frac{1}{2}$. The generators of $\mathfrak{sl}(2,\mathbb{R})$ act on this representation as
\begin{equation}
\mathcal{J}_0:= J_0\otimes\mathbbm{1} + \mathbbm{1}\otimes J_0,\quad \mathcal{J}_\pm:= J_\pm\otimes\mathbbm{1} + \mathbbm{1}\otimes J_\pm,
\end{equation}
with total Casimir
\begin{equation}\label{eq:coupling_casimir}
\mathcal{Q}:=-\mathcal{J}_0(\mathcal{J}_0+1) + \mathcal{J}_-\mathcal{J}_+ \equiv -\mathcal{J}_0(\mathcal{J}_0-1) + \mathcal{J}_+\mathcal{J}_-.
\end{equation}
\begin{remark}\label{rem:discrete+-}
The discrete negative representation $D^-_j$ is the \emph{dual representation} to $D^+_j$, i.e. they are related by the change
\begin{equation*}
J_0\rightarrow -J_0,\quad J_\pm \rightarrow -J_\mp,\quad \ket{j,m}\rightarrow (-1)^m\ket{j,-m}.
\end{equation*}
Conversely, $F_\gamma$ is dual to itself, i.e. it remains unchanged under the same change.
For this reason, the results in this section will be proved for $D^+_j$ only: the analogues for the negative representation trivially follow by transforming operators and vectors for both the finite and the discrete series.
\end{remark}
Such a representation is not generally irreducible. One is interested in finding, if it exists, the decomposition of $F_\gamma\otimes D^+_j$ in terms of irreducible representations of $\mathfrak{sl}(2,\mathbb{R})$, a non-trivial task since the representation is not unitary. Algebraically, this is equivalent to diagonalizing (if possible) the Casimir $\mathcal{Q}$.
Solving the eigenvalue equation for generic $\gamma$ is not easy; instead, the approach will be to explicitly find the eigenvectors and then show that, under certain conditions, they provide a basis for the product space.

To avoid confusion, the basis elements of the finite-dimensional series will be denoted by
\begin{equation}
\ket{\gamma,\mu},\quad \mu\in\left\{ -\gamma,\dotsc,\gamma\right\}
\end{equation}
from now on. Since both $F_\gamma$ and $D^+_j$ are lowest weight representations, i.e. $J_-$ annihilates one of their basis elements, their tensor product has to be as well. In fact, the vector
\begin{equation}
\ket{\psi_{(-\gamma)}}:=\ket{\gamma,-\gamma}\otimes\ket{j,j+1}
\end{equation}
satisfies
\begin{equation}
\mathcal{J}_-\ket{\psi_{(-\gamma)}}=0.
\end{equation}
An element of $F_\gamma\otimes D^+_j$ satisfying this property will be called a \emph{lowest weight vector}. $\ket{\psi_{(-\gamma)}}$ is trivially a $\mathcal{Q}$-eigenvector: from (\ref{eq:coupling_casimir}) follows that
\begin{equation}
\mathcal{Q}\ket{\psi_{(-\gamma)}}=-\mathcal{J}_0(\mathcal{J}_0 -1)\ket{\psi_{(-\gamma)}}=-(j-\gamma)(j-\gamma +1)\ket{\psi_{(-\gamma)}},
\end{equation}
since
\begin{equation}
\mathcal{J}_0\ket{\gamma,\mu}\otimes\ket{j,m}=(m+\mu)\ket{\gamma,\mu}\otimes\ket{j,m}.
\end{equation}

This is not the only lowest weight vector; in fact, one has

\begin{proposition}\label{prop:lowest_weight}
For the coupling $F_\gamma\otimes D^+_j$, the vectors
\begin{equation*}
\ket{\psi_{(\mu)}}=\sum_{\nu=-\gamma}^{\mu}(-1)^{\gamma+\nu}\prod_{\sigma=-\gamma}^{\nu-1}\frac{C_+(j,j+\mu-\sigma)}{C_+(\gamma,\sigma)}\ket{\gamma,\,\nu}\otimes\ket{j,\,j+1+\mu-\nu},
\end{equation*}
with $\mu\in\left\{ -\gamma,\dotsc,\gamma\right\}$ are lowest weight vectors and $\mathcal{Q}$-eigenvectors, with respective eigenvalues
\begin{equation*}
q_{(\mu)}:=-(j+\mu)(j+\mu+1).
\end{equation*}
\end{proposition}
\begin{proof}
First notice that each $\ket{\psi_{(\mu)}}$ is non-vanishing. Acting with $\mathcal{J}_-$, one gets
\begin{equation}
\begin{split}
\mathcal{J}_-\ket{\psi_{(\mu)}}=&\sum_{\nu=-\gamma}^{\mu}(-1)^{\gamma+\nu}\prod_{\sigma=-\gamma}^{\nu-1}\frac{C_+(j,j+\mu-\sigma)}{C_+(\gamma,\sigma)}C_+(\gamma,\nu-1)\ket{\gamma,\,\nu-1}\otimes\ket{j,\,j+1+\mu-\nu}+\\
&\sum_{\nu=-\gamma}^{\mu}(-1)^{\gamma+\nu}\prod_{\sigma=-\gamma}^{\nu-1}\frac{C_+(j,j+\mu-\sigma)}{C_+(\gamma,\sigma)}C_+(j,\,j+\mu-\nu)\ket{\gamma,\,\nu}\otimes\ket{j,\,j+\mu-\nu},
\end{split}
\end{equation}
where the property
\begin{equation}\label{eq:C+-}
C_+(j,m-1)=C_-(j,m)
\end{equation}
was used. Relabelling the dummy index $\nu$ in the first sum and noticing that the term $\nu=\mu$ vanishes in the second one, this can be rewritten as
\begin{equation}
\sum_{\nu=-\gamma}^{\mu-1}\left[(-1)^{\gamma+\nu+1}+(-1)^{\gamma+\nu}\right]\prod_{\sigma=-\gamma}^{\nu-1}\frac{C_+(j,j+\mu-\sigma)}{C_+(\gamma,\sigma)}
C_+(j,\,j+\mu-\nu)\ket{\gamma,\,\nu}\otimes\ket{j,\,j+\mu-\nu}=0.
\end{equation}
Again, the action of the Casimir is trivially given by
\begin{equation}
\mathcal{Q}\ket{\psi_{(\mu)}}=-\mathcal{J}_0(\mathcal{J}_0 -1)\ket{\psi_{(\mu)}}=-(j+\mu)(j+\mu +1)\ket{\psi_{(\mu)}}.
\end{equation}
\end{proof}

The fact that a finite number of eigenvectors exist does not mean $\mathcal{Q}$ is diagonalizable. Instead of working in an infinite-dimensional setting, however, one can take advantage of the tensor product basis vectors of $F_\gamma\otimes D^+_j$ being $\mathcal{J}_{0}$-eigenvectors: the space can be decomposed as
\begin{equation}
F_\gamma\otimes D^+_j=\bigoplus_{M=j+1-\gamma}^\infty V_M,
\end{equation}
where the $V_M$ are the orthogonal subspaces spanned by
\begin{equation}
\ket{(\mu)M}:=\ket{\gamma,\mu}\otimes\ket{j,M-\mu},\quad \mu\in\left\{-\gamma,\dotsc,\min(\gamma, M-j-1) \right\}.
\end{equation}
Each $V_M$ is finite-dimensional and, since $[\mathcal{Q},\mathcal{J}_0]=0$, one can work with the restriction $\mathcal{Q}_M:=\mathcal{Q}|_{V_M}$, satisfying
\begin{equation}
\mathcal{Q}_M(V_M)\subseteq V_M.
\end{equation}
The total Casimir $\mathcal{Q}$ will be diagonalizable if and only if each $\mathcal{Q}_M$ is, with eigenvalues not depending on $M$ and such that, for each $M$, the eigenvalues of $\mathcal{Q}_M$ are also eigenvalues of $\mathcal{Q}_{M+1}$.

In order to prove whether $\mathcal{Q}$ is diagonalizable or not and under which conditions, the following two lemmas will be needed.

\begin{lemma}\label{lem:repeated_J+}
If $j>\gamma-1$, then the repeated action of $\mathcal{J}_+$ on a lowest weight vector never vanishes; that is, for every $\mu$,
\begin{equation*}
(\mathcal{J}_+)^n\ket{\psi_{(\mu)}}\neq0\quad\forall n\in\mathbb{N}.
\end{equation*}
\end{lemma}
\begin{proof}
Suppose the lemma is not true for an arbitrary $\mu$, and let $n\geq 1$ be the smallest integer such that
\begin{equation}
(\mathcal{J}_+)^n\ket{\psi_{(\mu)}}=0.
\end{equation}
One has $(\mathcal{J}_+)^{n-1}\ket{\psi_{(\mu)}}\neq 0$ and, since $\mathcal{Q}$ and $\mathcal{J}_+$ commute,
\begin{equation}
\mathcal{Q}(\mathcal{J}_+)^{n-1}\ket{\psi_{(\mu)}}=(\mathcal{J}_+)^{n-1}\mathcal{Q}\ket{\psi_{(\mu)}}=q_{(\mu)}(\mathcal{J}_+)^{n-1}\ket{\psi_{(\mu)}}.
\end{equation}
On the other hand
\begin{equation}
\mathcal{Q}(\mathcal{J}_+)^{n-1}\ket{\psi_{(\mu)}}=-\mathcal{J}_0(\mathcal{J}_0 + 1)(\mathcal{J}_+)^{n-1}\ket{\psi_{(\mu)}}+\mathcal{J_-}(\mathcal{J}_+)^{n}\ket{\psi_{(\mu)}}=q_{(\mu+n)}(\mathcal{J}_+)^{n-1}\ket{\psi_{(\mu)}},
\end{equation}
since
\begin{equation}
(\mathcal{J}_+)^{n-1}\ket{\psi_{(\mu)}}\in V_{j+\mu+n}.
\end{equation}
This is only possible if $q_{(\mu)}=q_{(\mu+n)}$, that is
\begin{equation}
(j+\mu)(j+\mu+1)=(j+\mu+n)(j+\mu+n+1),
\end{equation}
which is equivalent to
\begin{equation}
n(n+2j+2\mu+1)=0.
\end{equation}
However, since $\mu\geq-\gamma$ and $j>\gamma-1$, one has
\begin{equation}
\begin{cases}
n\geq 1\\
n+2j+2\mu+1> 1+2(\gamma - 1) -2\gamma +1= 0,
\end{cases}
\end{equation}
which leads to a contradiction.
\end{proof}

\begin{lemma}\label{lem:distinct_eigenvalues}
The values
\begin{equation*}
q_{(\mu)}=-(j+\mu)(j+\mu+1),\quad \mu\in\left\{ -\gamma,\dotsc,\gamma\right\},\quad j\in\mathbb{C}
\end{equation*}
%are all distinct if and only if $j\not\in\mathbb{Z}/2$ or, when $j$ is (half-)integer, if $j>\gamma-1$ or $j<-\gamma$.
are all distinct if and only if
\begin{equation*}
j\not\in\mathbb{Z}/2\quad\mbox{or}\quad
\begin{cases}
j\in\mathbb{Z}/2\\
j\in (-\infty,-\gamma)\cup(\gamma-1,\infty).
\end{cases}
\end{equation*}
\end{lemma}
\begin{proof}
Consider arbitrary $\mu\neq\nu$. One can easily check that
\begin{equation}
q_{(\mu)}=q_{(\nu)}\Leftrightarrow (\mu-\nu)(\mu+\nu +2j+1)=0.
\end{equation}
Since $\mu$ and $\nu$ are different, this is equivalent to solving
\begin{equation}
\mu+\nu = -2j-1.
\end{equation}
The l.h.s. is an integer number, so if $j\not\in\mathbb{Z}/2$ there is no solution, i.e. the $q_{\mu}$'s are all different.

Suppose now that $j\in\mathbb{Z}/2$. The l.h.s. is subject to the constraint (remember $\mu\neq\nu$)
\begin{equation}
|\mu+\nu|<2\gamma,
\end{equation}
so that a solution exists if and only if
\begin{equation}
|2j+1|<2\gamma.
\end{equation}
Since $j$ can only change by half-integer steps, it follows that coinciding $q_{(\mu)}$'s exist if and only if $j\leq \gamma - 1$ and $j\geq-\gamma$. Consequently, they are all different if and only if $j>\gamma-1$ or $j<-\gamma$.
\end{proof}

It is now possible to prove the diagonalizability of $\mathcal{Q}$. One has the positive result:
\begin{proposition}\label{prop:discrete_diagonalizable}
When $j> \gamma -1$, the operator $\mathcal{Q}_M$ is diagonalizable, with distinct eigenvalues
\begin{equation*}
q_{(\mu)}=-(j+\mu)(j+\mu+1),\quad \mu\in\left\{ -\gamma,\dotsc,\min(\gamma,M-j-1)\right\}
\end{equation*}
independent of $M$.
\end{proposition}
\begin{proof}
Define, up to a normalization factor, the vectors
\begin{equation}
\ket{j+\mu,M}:=(\mathcal{J}_+)^{M-j-1-\mu}\ket{\psi_{(\mu)}}\in V_M,\quad \mu\in\left\{ -\gamma,\dotsc,\min(\gamma,M-j-1)\right\};
\end{equation}
owing to Lemma \ref{lem:repeated_J+}, they are all non-vanishing. 
Moreover, since $\mathcal{Q}$ commutes with $\mathcal{J}_+$, they are $\mathcal{Q}_M$-eigenvectors, with eigenvalues $q_{(\mu)}$. Finally, it follows from Lemma \ref{lem:distinct_eigenvalues} that the eigenvalues are all distinct: since the number of eigenvalues equals the dimension of $V_M$, $\mathcal{Q}_M$ is diagonalizable.
\end{proof}
As a consequence, the total Casimir $\mathcal{Q}$ will be overall diagonalizable. Conversely, one can prove:
\begin{proposition}\label{prop:discrete_non-diagonalizable}
When $j\leq \gamma -1$, the operator $\mathcal{Q}_{j+1+\gamma}$ is not diagonalizable.
\end{proposition}
\begin{proof}
The proof is divided in two parts: first one shows that the only possible eigenvalues of $\mathcal{Q}_{j+1+\gamma}$ are the $q_{(\mu)}$'s. This will then be used to show that $\mathcal{Q}_{j+1+\gamma}$ is not diagonalizable.

Suppose there is a non-zero eigenvector $\ket{\varphi}\in V_{j+1+\gamma}$, with eigenvalue
\begin{equation}
\varphi\neq q_{(\mu)},\quad \mu \in \left\{-\gamma,\dotsc,\gamma\right\}.
\end{equation}
It must be
\begin{equation}
(\mathcal{J}_-)^n\ket{\varphi}=0
\end{equation}
for some
\begin{equation}
n\in\left\{1,2,\dotsc,2\gamma+1\right\},
\end{equation}
since there is only one vector in $V_{j+1-\gamma}$ and it is annihilated by $\mathcal{J}_-$. Let $N$ be the smallest such number; then $(\mathcal{J}_-)^{N-1}\ket{\varphi}\neq 0$ and
\begin{equation}
\mathcal{Q}(\mathcal{J}_-)^{N-1}\ket{\varphi}=(\mathcal{J}_-)^{N-1}\mathcal{Q}\ket{\varphi}=\varphi(\mathcal{J}_-)^{N-1}\ket{\varphi},
\end{equation}
while at the same time
\begin{equation}
\mathcal{Q}(\mathcal{J}_-)^{N-1}\ket{\varphi}=-\mathcal{J}_0(\mathcal{J}_0 -1)(\mathcal{J}_-)^{N-1}\ket{\varphi}+ \mathcal{J}_+(\mathcal{J}_-)^N\ket{\varphi}=q_{(\gamma-N+1)}(\mathcal{J}_-)^{N-1}\ket{\varphi}.
\end{equation}
It follows that $\varphi$ equals one of the $q_{(\mu)}$'s, which is a contradiction. This concludes the first part of the proof.

For the second part notice that, since $j\geq -\tfrac{1}{2}$, it is always $j\geq -\gamma$. Then, since $j\leq \gamma - 1$, it follows from Lemma \ref{lem:distinct_eigenvalues} that there are at most $2\gamma$ distinct eigenvalues. However, by acting with $\mathcal{Q}_{j+1+\gamma}$ on the basis vectors
\begin{equation}
\ket{(\mu)j+1+\gamma}=\ket{\gamma,\mu}\otimes\ket{j,j+1+\gamma-\mu}\in V_{j+1+\gamma},
\end{equation}
one obtains that the matrix elements
\begin{equation}
\mathcal{Q}_{\mu\nu}:=\braket{(\mu)j+1+\gamma|\mathcal{Q}|(\nu)j+1+\gamma}
\end{equation}
are non-vanishing only if
\begin{equation}
\mu=\nu\quad\mbox{or}\quad \mu=\nu\pm 1;
\end{equation}
in other words, $\mathcal{Q}_{\mu\nu}$ are the entries of a \emph{tridiagonal matrix} (see appendix \ref{app:tridiagonal}). In particular, it follows from Corollary \ref{coroll:tridiagonal_eigenspaces} that the eigenspaces of a tridiagonal matrix (or of an operator represented by such a matrix in a particular basis) are all $1$-dimensional. As a consequence, there are at most $2\gamma$ eigenvectors, which means $\mathcal{Q}_{j+1+\gamma}$ is not diagonalizable.
\end{proof}
Since in this case $\mathcal{Q}_M$ is non-diagonalizable for at least one $M$, $\mathcal{Q}$ will not be diagonalizable. To summarize, the coupling $F_\gamma\otimes D^+_j$ can be decomposed in irreducible representations if and only if $j>\gamma-1$.

An eigenbasis for $\mathcal{Q}$ can be constructed by defining recursively
\begin{equation}
\ket{J,M+1}=\frac{1}{C_+(J,M)}\mathcal{J}_+\ket{J,M},\quad J\in\left\{ j-\gamma,\dotsc,j+\gamma\right\}
\end{equation}
with
\begin{equation}
\begin{cases}
\mathcal{Q}\ket{J,M}=q_{(J-j)}\ket{J,M}\\
\mathcal{J}_0\ket{J,M}=M\ket{J,M},
\end{cases}
\end{equation}
starting from
\begin{equation}
\ket{J,J+1}:=\ket{\psi_{J-j}}
\end{equation}
up to a normalization factor; Lemma \ref{lem:repeated_J+} guarantees that they are all non-zero. One can easily see that each $\mathcal{Q}$-eigenspace behaves as the discrete positive representation $D^+_J$.

In terms of the the old basis elements, the change of basis must be of the form
\begin{equation}
\ket{j+\mu,M}=\sum_{\nu=-\gamma}^{\Omega_M}A^M_{\nu\mu}(j,\gamma)\ket{(\nu)M},\quad \mu\in\left\{-\gamma,\dotsc,\Omega_M\right\},\quad \Omega_M:=\min(\gamma, M-j-1),
\end{equation}
with the $A^M_{\nu\mu}$'s forming an invertible matrix. They will be called \emph{Clebsch--Gordan coefficients}, in analogy with $\mathfrak{su}(2)$ representation theory.
\subsection{Coupling of finite and continuous representations}\label{sec:fin-cont}
Consider now the coupling $F_\gamma\otimes C^\varepsilon_j$ of a finite-dimensional representation and a generic one from the continuous series, not necessarily unitary. The technique used for the discrete series will not work here, because the spectrum of $J_0$ is unbounded, hence a different approach is needed.

Again, one can work individually on each $\mathcal{J}_0$-eigenspace $V_M$, with basis vectors
\begin{equation}
\ket{(\mu)M}=\ket{\gamma, \mu}\otimes \ket{j,M-\mu},\quad \mu\in\left\{-\gamma,\dotsc,\gamma \right\},
\end{equation}
and try to diagonalize $\mathcal{Q}_M$. Explicitly, one is interested in finding a change of basis
\begin{equation}
\ket{J_{(\mu)},M}=\sum_{\nu=-\gamma}^\gamma A^M_{\nu\mu}(j,\gamma)\ket{(\nu)M},\quad \mu\in\left\{-\gamma,\dots,\gamma\right\},
\end{equation}
with
\begin{equation}
\mathcal{Q}\ket{J_{(\mu)},M}=-J_{(\mu)}\left(J_{(\mu)} +1\right)\ket{J_{(\mu)},M}.
\end{equation}
\begin{remark}
Since any non-trivial $F_\gamma$ is not unitary, the total Casimir is not Hermitian; moreover, one can easily check that it is not a \emph{normal} operator either, i.e.
\begin{equation*}
[Q^\dagger_M,Q_M]\neq 0.
\end{equation*}
As a consequence, not only the spectral theorem cannot be used to diagonalize it, but its eigenvectors will be non-orthogonal and the matrix $A^M(j,\gamma)$ non-unitary.
\end{remark}
Solving the eigenvalue equation explicitly for arbitrary $\gamma$ is too difficult. However, one can easily do it for the $2$-dimensional case $\gamma=\frac{1}{2}$: each $\mathcal{Q}_M$ is diagonalizable if and only if $j\neq-\frac{1}{2}$, with eigenvalues $q_{(\pm\frac{1}{2})}$ (the corresponding Clebsch--Gordan coefficients are listed in appendix \ref{app:CGtable}). Using this information, one can prove by induction that, when $j\not\in\mathbb{Z}/2$, $\mathcal{Q}$ is diagonalizable for all $\gamma\geq\frac{1}{2}$. The case $j\in\mathbb{Z}/2$ will be treated later with a different method.
\begin{proposition}\label{prop:continuous_non_Z/2}
When $j\not\in\mathbb{Z}/2$, the eigenvalues of $\mathcal{Q}_M$ are
\begin{equation*}
q_{(\mu)}=-(j+\mu)(j+\mu+1),\quad \mu\in\left\{ -\gamma,\dotsc,\gamma\right\},
\end{equation*}
that is
\begin{equation*}
J_{(\mu)}=j+\mu.
\end{equation*}
These are all distinct, so $\mathcal{Q}$ is diagonalizable.
\end{proposition}
\begin{proof}
The proof proceeds by induction on half-integer $\gamma\geq\frac{1}{2}$. The statement is true for $\gamma=\frac{1}{2}$: suppose that it is true for $\gamma-\frac{1}{2}$ and consider the coupling $F_\gamma\otimes C^\varepsilon_j$.

The finite-dimensional representations are isomorphic to the unitary representations of~$\mathfrak{su}(2)$, seen as representations of the complexification $\mathfrak{sl}(2,\mathbb{R})_\mathbb{C}\cong \mathfrak{su}(2)_\mathbb{C}$. Consequently, the well-known result of $\mathfrak{su}(2)$ recoupling theory \cite{barut}
\begin{equation}
F_\gamma \subset  F_{\frac{1}{2}}\otimes F_{\gamma-\frac{1}{2}}\cong F_{\gamma-1}\oplus F_\gamma
\end{equation}
can be used. Explicitly
\begin{equation}
\ket{\gamma,\mu}\equiv \sum_{\sigma=-\frac{1}{2}}^\frac{1}{2}\sum_{\lambda=-\gamma+\frac{1}{2}}^{\gamma-\frac{1}{2}} \braket{  \tfrac{1}{2},\sigma ;\gamma-\tfrac{1}{2},\lambda | \gamma,\mu}  \ket{\tfrac{1}{2},\sigma}\otimes \ket{\gamma-\tfrac{1}{2},\lambda},
\end{equation}
where
\begin{equation}
\braket{  \tfrac{1}{2},\sigma ;\gamma-\tfrac{1}{2},\lambda | \gamma,\mu}
\end{equation}
are the $\mathfrak{su}(2)$ Clebsch--Gordan coefficients. One can then write, since $F_{\gamma-\frac{1}{2}}\otimes C^\varepsilon_j$ is decomposable by induction hypothesis,
\begin{align}
\ket{\gamma,\mu}\otimes\ket{j,M-\mu}=&\sum_{\sigma,\lambda} \braket{  \tfrac{1}{2},\sigma ;\gamma-\tfrac{1}{2},\lambda | \gamma,\mu} \ket{\tfrac{1}{2},\sigma} \otimes \biggl[\ket{\gamma-\tfrac{1}{2},\lambda}\otimes\ket{j,M-\mu}\biggr]\\
=&\sum_{\sigma,\lambda} \braket{  \tfrac{1}{2},\sigma ;\gamma-\tfrac{1}{2},\lambda | \gamma,\mu} \ket{\tfrac{1}{2},\sigma} \otimes\sum_{\kappa=-\gamma+\frac{1}{2}}^{\gamma-\frac{1}{2}}B^{M-\sigma}_{\kappa\lambda}\left(j,\gamma-\tfrac{1}{2}\right) \ket{j+\kappa,M-\sigma},
\end{align}
where the $B^M_{\kappa\lambda}$ are the \emph{inverse} Clebsch--Gordan coefficients, i.e. $B^M$ is the inverse of the matrix $A^M$.

In particular, when $\mu=-\gamma$, the only non-zero $\mathfrak{su}(2)$ coefficient is\footnote{The Condon--Shortley convention \cite{condon_shortley} is used here.}
\begin{equation}
\braket{ \tfrac{1}{2},-\tfrac{1}{2} ; \gamma-\tfrac{1}{2},-\gamma+\tfrac{1}{2} | \gamma,-\gamma}=1
\end{equation}
so that
\begin{equation}\label{eq:-gamma}
\ket{\gamma,-\gamma}\otimes\ket{j,M+\gamma}
=\sum_{\rho=-\frac{1}{2}}^\frac{1}{2}\sum_{\kappa=-\gamma+\frac{1}{2}}^{\gamma-\frac{1}{2}}B^{M+\frac{1}{2}}_{\kappa,-\gamma+\frac{1}{2}}(j,\gamma-\tfrac{1}{2})B^{M}_{\rho,-\tfrac{1}{2}}\left(j+\kappa,\tfrac{1}{2}\right)\ket{\left(j+\kappa\right)j+\rho+\kappa,M},
\end{equation}
where the $(j+\kappa)$ label in the vector indicates it comes from the coupling
\begin{equation}
\ket{\tfrac{1}{2},-\tfrac{1}{2} }\otimes\ket{j+\kappa,M+\tfrac{1}{2}}.
\end{equation}
There are exactly $4\gamma$ vectors on the r.h.s of eq. (\ref{eq:-gamma}): they are
\begin{equation}
\begin{cases}
\ket{\left(j-\gamma+\tfrac{1}{2}\right)j-\gamma,M}\\
\ket{\left(j+\mu-\tfrac{1}{2}\right)j+\mu,M}\quad\mbox{and}\quad\ket{\left(j+\mu+\tfrac{1}{2}\right)j+\mu,M},\qquad \mu\in\left\{-\gamma+1,\dotsc,\gamma-1\right\}\\
\ket{\left(j+\gamma-\tfrac{1}{2}\right)j+\gamma,M}.
\end{cases}
\end{equation}
Their $\mathcal{Q}$-eigenvalues are
\begin{equation}
q_{(\mu)}=-(j+\mu)(j+\mu+1),\quad \mu\in\left\{ -\gamma,\dotsc,\gamma\right\},
\end{equation}
which are all distinct (see Lemma \ref{lem:distinct_eigenvalues}), and they form a basis for the $M$ eigenspace in $ V_\frac{1}{2}\otimes V_{\gamma-\frac{1}{2}} \otimes C_{j}^\varepsilon$, i.e. they are independent.

As shown in appendix \ref{app:CGprop}, Clebsch--Gordan coefficients satisfy the property
\begin{equation}
B_{\nu+1,-\gamma}^M(j,\gamma)=\alpha_\nu(j,\gamma)\frac{\sqrt{j+\nu-M +1}}{\sqrt{j+\nu+M+1}} B_{\nu,-\gamma}^M(j,\gamma),
\end{equation}
where $\alpha_\nu$ is fixed by the normalization convention and does not depend on $M$. Using this formula and the fact that (see appendix \ref{app:CGtable})
\begin{equation}
B^M_{\rho,-\frac{1}{2}}\left(j+\kappa,\tfrac{1}{2}\right)=
\begin{cases}
-\frac{\sqrt{j+\kappa+M+\frac{1}{2}}}{\sqrt{2j+2\kappa+1}}\qquad&\mbox{if}\quad\rho=-\tfrac{1}{2}\\
\frac{\sqrt{j+\kappa-M+\frac{1}{2}}}{\sqrt{2j+2\kappa+1}}\qquad&\mbox{if}\quad\rho=\tfrac{1}{2},
\end{cases}
\end{equation}
it is possible to write
\begin{equation}
\ket{(-\gamma)M}=\ket{\gamma,-\gamma}\otimes\ket{j,M+\gamma}=\sum_{\nu=-\gamma}^\gamma B^M_{\nu,-\gamma}(j,\gamma)\ket{j+\nu,M}
\end{equation}
for some coefficients $B^M_{\nu,-\gamma}$, where the vectors on the r.h.s are defined up to a normalization factor as
\begin{equation}
\ket{J,M}:=
\begin{cases}
\ket{\left(J+\tfrac{1}{2}\right)J,M} &\qquad\mbox{if}\quad J=j-\gamma\\[0.1cm]
\ket{\left(J-\tfrac{1}{2}\right)J,M} &\qquad\mbox{if}\quad J=j+\gamma\\[0.1cm]
\frac{1}{\sqrt{2J}}\ket{\left(J-\tfrac{1}{2}\right)J,M}-\frac{\beta(J)}{\sqrt{2J+2}}\ket{\left(J+\tfrac{1}{2}\right)J,M}&\qquad\mbox{otherwise},
\end{cases}
\end{equation}
with
\begin{equation}
\beta\left(j+\kappa+\tfrac{1}{2}\right)=\alpha_\kappa\left(j,\gamma-\tfrac{1}{2}\right).
\end{equation}
Since these vectors live in different $\mathcal{Q}$-eigenspaces, they are necessarily independent. Suppose the following is true:
\begin{equation}\label{eq:nested}
\ket{J,M}\in V_M,\quad \forall J\in\left\{j-\gamma,\dotsc,j+\gamma\right\}.
\end{equation}
Then they would be $2\gamma+1$ independent eigenvectors in $V_M$, i.e. an eigenbasis, which proves the proposition. It only remains to show that (\ref{eq:nested}) is indeed true; this can be done by induction as well.

It is easily checked that, for $\mu<\gamma$,
\begin{equation}
\mathcal{J}_-\mathcal{J}_+\ket{(\mu)M}\in\mathrm{span}\left\{\ket{(\mu-1)M},\ket{(\mu)M},\ket{(\mu+1)M}\right\},
\end{equation}
with
\begin{equation}
\braket{(\mu+1)M|\mathcal{J}_-\mathcal{J}_+|(\mu)M}=C_+(\gamma,\mu)C_-(j,M-\mu)\neq 0.
\end{equation}
Consequently, it must be
\begin{equation}
\ket{(\mu+1)M}\in\mathrm{span}\left\{\ket{(\mu-1)M},\ket{(\mu)M},\mathcal{J}_-\mathcal{J}_+\ket{(\mu)M}\right\}.
\end{equation}
Suppose that
\begin{equation}
\ket{(\mu-1)M},\ket{(\mu)M}\in\mathrm{span}\left\{\ket{J,M}|J=j-\gamma,\dotsc,j+\gamma\right\}.
\end{equation}
Since $\mathcal{J}_-\mathcal{J}_+$ commutes with $\mathcal{Q}_M$, one has
\begin{equation}
\mathcal{J}_-\mathcal{J}_+\ket{(\mu)M}\in\mathrm{span}\left\{\ket{J,M}|J=j-\gamma,\dotsc,j+\gamma\right\},
\end{equation}
so that
\begin{equation}
\ket{(\mu+1)M}\in\mathrm{span}\left\{\ket{J,M}|J=j-\gamma,\dotsc,j+\gamma\right\}
\end{equation}
as well. 
Since the hypothesis is valid for $\mu=-\gamma$ (note that $\ket{(-\gamma-1)M}\equiv\vec{0}$), it follows by induction that every basis vector $\ket{(\mu)M}$ can be written as a linear combination of the independent $\ket{J,M}$ vectors. As their number match, the latter must form a basis for $V_M$, so that they are, in fact, eigenvectors for $\mathcal{Q}_M$.
\end{proof}
When $j\in\mathbb{Z}/2$, $\mathcal{Q}_M$ is not always diagonalizable. In order to prove when it can be done, the following lemma is needed.
\begin{lemma}\label{lem:continuous_Z/2}
When $j\in\mathbb{Z}/2$, the eigenvalues of $\mathcal{Q}_M$ are given by
\begin{equation*}
q_{(\mu)}=-(j+\mu)(j+\mu+1),\quad \mu \in \left\{ -\gamma,\dotsc,\gamma\right\}.
\end{equation*}
\end{lemma}
\begin{proof}
The result follows by continuity from Proposition \ref{prop:continuous_non_Z/2}. First notice that the function
\begin{equation}
d(j,\lambda):=\det(\mathcal{Q}_M-\lambda\mathbbm{1})
\end{equation}
is continuous (in the complex plane) in $j\in\mathbb{R}$, since it is a product of continuous functions of $j$. Moreover, for $j\not\in\mathbb{Z}/2$, it is given by
\begin{equation}
d(j,\lambda)=\prod_{\mu=-\gamma}^\gamma[-(j+\mu)(j+\mu+1)-\lambda],
\end{equation}
as a consequence of Proposition \ref{prop:continuous_non_Z/2}. Now let $k\in\mathbb{Z}/2$; since $d$ is continuous, it must be
\begin{equation}
d(k,\lambda)=\lim_{j\rightarrow k} d(j,\lambda)=\prod_{\mu=-\gamma}^\gamma[-(k+\mu)(k+\mu+1)-\lambda]
\end{equation}
so that the eigenvalues of $\mathcal{Q}_M$ are the $q_{(\mu)}$'s.
\end{proof}
It is now possible to prove that
\begin{proposition}\label{prop:continuous_Z/2}
When $j\in\mathbb{Z}/2$, $\mathcal{Q}$ is diagonalizable if and only if $j>\gamma-1$ or $j<-\gamma$.
\end{proposition}
\begin{proof}
One has from Lemma \ref{lem:distinct_eigenvalues} from the previous section that the eigenvalues of each $\mathcal{Q}_M$ (given by Lemma \ref{lem:continuous_Z/2}) are all distinct if and only if $j>\gamma-1$ or $j<-\gamma$. However, like in the discrete case (see proof of Proposition \ref{prop:discrete_non-diagonalizable}), $\mathcal{Q}_M$ is represented in the $\ket{(\mu)M}$ basis by a tridiagonal matrix. It follows from Corollary \ref{coroll:tridiagonal_eigenspaces} that the $\mathcal{Q}_M$ are diagonalizable if and only if the eigenvalues are all different, i.e. $j>\gamma-1$ or $j<-\gamma$, as required.
\end{proof}

To summarize, the coupling $F_\gamma\otimes C^\varepsilon_j$ can be decomposed in irreducible representations if and only if $j\not\in\mathbb{Z}/2$ or, when $j$ is (half-)integer, if $j>\gamma-1$ or $j<-\gamma$. One can check directly that each $\mathcal{Q}$-eigenspace behaves as a continuous representation.

\subsection{Wigner--Eckart theorem for \texorpdfstring{$\mathfrak{sl}(2,\mathbb{R})$}{sl(2,R)}}\label{sec:WE}
Before stating the theorem, some new notation will be defined.
Let $\rho_j$ be a generic irreducible representation on the vector space $V_j$, where $j$ is to be thought as encoding, in addition to its numerical value, class and parity of the representation.
It will be useful to define
\begin{equation}
\mathcal{D}(\gamma,j):=
\left\{
j'\left\vert\right. \rho_{j'}\subseteq F_\gamma\otimes\rho_j
\right\},
\end{equation}
i.e. the set of labels $j'$ appearing in the decomposition of $F_\gamma\otimes\rho_j$ (assuming such a decomposition exists). Moreover, $\mathcal{M}_j$ will denote the set of possible $m$ values of the representation $\rho_j$, i.e.
\begin{equation}
\mathcal{M}_j:=
\left\{
m\in\mathbb{Z}/2 \left\vert\right. \ket{j,m}\in V_j
\right\}.
\end{equation}
Lastly, instead of working with the matrix form of the Clebsch--Gordan coefficients, one can define
\begin{subequations}\label{eq:CGs}
\begin{align}
\label{eq:CG}
A(\gamma,\mu;j,m|J,M):=&\sum_{\nu=-\gamma}^\gamma A^{M}_{\mu\nu}(j,\gamma)\,\delta_{M,m+\mu}\,\delta_{J,j+\nu}\\
\label{eq:inverse_CG}
B(J,M|\gamma,\mu;j,m):=&\sum_{\nu=-\gamma}^\gamma B^{M}_{\nu\mu}(j,\gamma)\,\delta_{M,m+\mu}\,\delta_{J,j+\nu},
\end{align}
\end{subequations}
which will still be referred to as Clebsch--Gordan coefficients.
As $A^M$ and $B^M$ are one the inverse of the other, the new coefficients satisfy the orthogonality relations
\begin{subequations}
\begin{equation}\label{eq:CG_ortho1}
\sum_{J\in\mathcal{D}(\gamma,j)}\sum_{M\in\mathcal{M}_{J}} A(\gamma,\mu;j,m|J,M)B(J,M|\gamma,\mu';j,m')=\delta_{\mu,\mu'}\,\delta_{m,m'},\quad \mu\in\mathcal{M}_\gamma,\quad m\in\mathcal{M}_j
\end{equation}
\begin{equation}\label{eq:CG_ortho2}
\sum_{\mu\in\mathcal{M}_\gamma}\sum_{m\in\mathcal{M}_j} B(J,M|\gamma,\mu;j,m)A(\gamma,\mu;j,m|J',M')=\delta_{J,J'}\,\delta_{M,M'},\quad J\in\mathcal{D}(\gamma,j),\quad M\in\mathcal{M}_J.
\end{equation}
\end{subequations}
The finite-dimensional case can be covered as well by putting
\begin{equation}
A(\gamma,\mu;j,m|J,M):=\braket{\gamma,\mu;j,m|J,M},\quad B(J,M|\gamma,\mu;j,m):=\braket{J,M|\gamma,\mu;j,m},
\end{equation}
which satisfy the same orthogonality relations \cite{messiah2}.

It is now possible to prove the Wigner--Eckart theorem.
\begin{theorem*}[Wigner--Eckart for $\mathfrak{sl}(2,\mathbb{R})$]\label{thm:WE}
Let $T^\gamma$ be an $\sltr$ tensor operator of rank $\gamma$ between two irreducible representations $\rho_j$ and $\rho_{j'}$. When $F_\gamma\otimes \rho_j$ admits a decomposition in irreducible representations, the matrix elements of $T^\gamma$ can be expressed as
\begin{equation*}
\braket{j',m'|T^\gamma_\mu|j,m}=\braket{j'\Vert T^\gamma\Vert j}B(j',m'|\gamma,\mu;j,m),
\end{equation*}
where the \emph{reduced matrix element} $\braket{j'\Vert T^\gamma\Vert j}\in\mathbb{C}$ does not depend on $m$, $m'$ or $\mu$. In particular, if $\rho_{j'}$ is not in the decomposition of $F_\gamma\otimes \rho_j$, the matrix elements necessarily vanish.
\end{theorem*}
\begin{proof}
If $F_\gamma\otimes \rho_j$ admits a decomposition in irreducible representations, the Clebsch--Gordan coefficients (\ref{eq:CG}) exist, and one can define the vectors
\begin{equation}
\ket{\psi_{j'',m''}}:=\sum_{\mu\in\mathcal{M}_\gamma}\sum_{m\in\mathcal{M}_j}A(\gamma,\mu;j,m|j'',m'')\,T^\gamma_\mu\ket{j,m},\quad j''\in\mathcal{D}(\gamma,j),\quad m''\in\mathcal{M}_{j''}.
\end{equation}
By virtue of the orthogonality relation (\ref{eq:CG_ortho1}), this can be inverted to get
\begin{equation}\label{eq:WE_T_range}
T^\gamma_\mu\ket{j,m}=\sum_{j''\in\mathcal{D}(\gamma,j)}\sum_{m''\in\mathcal{M}_{j''}}B(j'',m''|\gamma,\mu;j,m)\ket{\psi_{j'',m''}}.
\end{equation}
Now consider the action of the $\sltr$ generators on the $\ket{\psi_{j'',m''}}$ vectors. One finds, using the definition of tensor operator (\ref{eq:tenso_op_def}), that
\begin{align}
J_0\ket{\psi_{j'',m''}}&=\sum_{\mu,m}A(\gamma,\mu;j,m|j'',m'')\,J_0T^\gamma_\mu\ket{j,m}\notag\\
&=\sum_{\mu,m}A(\gamma,\mu;j,m|j'',m'')\,\left\{T^\gamma_\mu J_0 +[J_0,T^\gamma_\mu]\right\}\ket{j,m}\notag\\
&=\sum_{\mu,m}A(\gamma,\mu;j,m|j'',m'')\,(m+\mu)\ket{j,m};
\end{align}
using the fact that the Clebsch--Gordan coefficients vanish unless $m+\mu=m''$, one has
\begin{equation}\label{eq:J0_psi_action}
J_0\ket{\psi_{j'',m''}}=m''\ket{\psi_{j'',m''}}.
\end{equation}
Analogously, one has
\begin{align}
J_\pm\ket{\psi_{j'',m''}}&=\sum_{\mu,m}A(\gamma,\mu;j,m|j'',m'')\,\left\{T^\gamma_\mu J_\pm +[J_\pm,T^\gamma_\mu]\right\}\ket{j,m}\notag\\
=&\sum_{\mu\in\mathcal{M}_\gamma}\sum_{m\in\mathcal{M}_j}A(\gamma,\mu;j,m|j'',m'')\,\left\{C_\pm(j,m)\,T^\gamma_\mu\ket{j,m\pm 1}+ C_\pm(\gamma,\mu)\,T_{\mu\pm 1}^\gamma\ket{j,m}\right\}\notag\\
=& \sum_{\mu\in\mathcal{M}^\pm_\gamma}\sum_{m\in\mathcal{M}_j}C_\pm(\gamma,\mu \mp 1)A(\gamma,\mu \mp 1;j,m|j'',m'')\,T^\gamma_\mu\ket{j,m} \notag\\
\label{eq:J+-_psi_action0}
& + \sum_{\mu\in\mathcal{M}_\gamma}\sum_{m\in\mathcal{M}^\pm_j}C_\pm(j,m\mp 1)A(\gamma,\mu;j,m\mp 1|j'',m'')\,T^\gamma_\mu\ket{j,m},
\end{align}
where
\begin{equation}
\mathcal{M}^\pm:=\left\{m\pm 1\left\vert\right. m\in\mathcal{M} \right\}.
\end{equation}
One can easily see that
\begin{equation}
C_\pm(\gamma,\mu)=0\quad\forall \mu\in \mathcal{M}_\gamma\Delta\mathcal{M}^\pm_\gamma := (\mathcal{M}_\gamma\setminus \mathcal{M}^\pm_\gamma)\cup (\mathcal{M}^\pm_\gamma \setminus \mathcal{M}_\gamma)
\end{equation}
and that, for arbitrary $\rho_j$, it is either
\begin{equation}
\mathcal{M}_j\Delta\mathcal{M}^\pm_j=\emptyset\quad\mbox{or}\quad C_\pm(j,m)=0\quad\forall m \in \mathcal{M}_j\Delta\mathcal{M}^\pm_j;
\end{equation}
One can then rewrite (\ref{eq:J+-_psi_action0}) as
\begin{equation}
\sum_{\mu\in\mathcal{M}_\gamma}\sum_{m\in\mathcal{M}_j}
\left\{
C_\pm(\gamma,\mu \mp 1)A(\gamma,\mu \mp 1;j,m|j'',m'') +
C_\pm(j,m\mp 1)A(\gamma,\mu;j,m\mp 1|j'',m'')
\right\}
T^\gamma_\mu\ket{j,m}.
\end{equation}
Using the Clebsch--Gordan recursion relation from appendix \ref{app:CGprop}, which reads
\begin{equation}
C_\pm(J,M)A(\gamma,\mu;j,m|J,M\pm 1)=
C_\pm(\gamma,\mu \mp 1)A(\gamma,\mu \mp 1;j,m|J,M) +
C_\pm(j,m\mp 1)A(\gamma,\mu;j,m\mp 1|J,M),
\end{equation}
one finally finds
\begin{equation}\label{eq:J+-_psi_action}
J_\pm\ket{\psi_{j'',m''}}=C_\pm(j'',m'')\ket{\psi_{j'',m''\pm1}}.
\end{equation}
Equations (\ref{eq:J0_psi_action}) and (\ref{eq:J+-_psi_action}) imply that
\begin{equation}
\mathrm{span}\left\{
\ket{\psi_{j'',m''}}
\left\vert\right.
m\in\mathcal{M}_{j''}\right\}
\cong
V_{j''}
\end{equation}
and
\begin{equation}
\ket{\psi_{j'',m''}}\propto \ket{j'',m''}.
\end{equation}
One can show that the proportionality factor
\begin{equation}
N(j'',m''):=\braket{j'',m''|\psi_{j'',m''}}
\end{equation}
does not depend on $m''$. In fact, consider the matrix element
\begin{equation}
\braket{j'',m''+1|J_+|\psi_{j'',m''}}=C_+(j'',m'')N(j'',m''+1).
\end{equation}
It can be rewritten as
\begin{equation}
\braket{j'',m''+1|J_+|\psi_{j'',m''}}=N(j'',m'')\braket{j'',m''+1|J_+|j'',m''}=C_+(j'',m'')N(j'',m''),
\end{equation}
so that
\begin{equation}
N(j'',m''+1)=N(j'',m''),\quad\forall m''\in\mathcal{M}_{j''},
\end{equation}
which is only possible if $N$ doesn't depend on $m''$.

Now, equation (\ref{eq:WE_T_range}) implies that the range of $T^\gamma$ is spanned by all the vectors $\ket{\psi_{j'',m''}}$; since it has to be a subset of $V_{j'}$, it must necessarily be
\begin{equation}
N(j'')=0\quad\forall j''\neq j'.
\end{equation}
The matrix elements of $T^\gamma$ are then given by
\begin{equation}
\braket{j',m'|T^\gamma_\mu|j,m}=N(j')B(j',m'|\gamma,\mu;j,m),
\end{equation}
which includes the case $j'\not\in\mathcal{D}(\gamma,j)$ as the Clebsch--Gordan coefficients vanish in this case. The theorem is recovered by putting
\begin{equation}
\braket{j''\Vert T^\gamma\Vert j}:= N(j'').
\end{equation}
This concludes the proof.
\end{proof}

\subsection{An application: the Jordan--Schwinger representation}
An application of the Wigner--Eckart theorem for $\sltr$ will be presented here. It is well known in the quantum theory of angular momentum, where the Lie algebra $\mathfrak{su}(2)$ is used, that the generators of the algebra (physically corresponding to infinitesimal rotations) can be expressed in terms of a pair of uncoupled \emph{quantum harmonic oscillators} \cite{schwinger1952}. This result is known as \emph{Jordan--Schwinger representation}. Explicitly, the generators $K_z$, $K_+$ and $K_-$~---~in the physicist convention, analogue to (\ref{eq:phys_generators})~---~with commutation relations
\begin{equation}
[K_z,K_\pm]=\pm K_\pm,\quad [K_+,K_-]=2K_z
\end{equation}
can be expressed as
\begin{equation}
K_z=\tfrac{1}{2}\left(a^\dagger a - b^\dagger b \right),\quad K_+=a^\dagger b,\quad K_-=b^\dagger a,
\end{equation}
where $a$ and $b$ are quantum harmonic oscillators, i.e. satisfy
\begin{equation}
[a,a^\dagger]=[b,b^\dagger]=\mathbbm{1},
\end{equation}
and all the other commutators vanish\footnote{More generally $a$, $a^\dagger$, $b$, $b^\dagger$ and $\mathbbm{1}$ form a unitary representation of the $5$-dimensional Heisenberg algebra $\mathfrak{h}_2(\mathbb{R})$.}

One may ask if a similar result holds for $\sltr$ representations: the answer is positive for the discrete and finite-dimensional series, but an analogous construction for the continuous series is not easily guessed and, in fact, was not available until now. It will be shown here how the Wigner--Eckart theorem can be used to find an analogous of the Jordan--Schwinger representation for $\sltr$, which covers all representation classes.

First notice that a rank-$1$ tensor operator $V$ can be constructed out of the algebra generators, with components
\begin{equation}
V_{\pm 1}= \mp i J_\pm,\quad V_0=-\sqrt{2}J_0.
\end{equation}
An alternative way to look at the Jordan--Schwinger construction is to look for two rank-$\frac{1}{2}$ tensor operators $T$ and $\widetilde{T}$ that can be combined to obtain $V$. Explicitly, one can make the ansatz
\begin{equation}\label{eq:ansatz1}
V_\mu=\sum_{\mu_1=-\frac{1}{2}}^\frac{1}{2} \sum_{\mu_2=-\frac{1}{2}}^\frac{1}{2} \braket{\tfrac{1}{2},\mu_1;\tfrac{1}{2},\mu_2|1,\mu}
\, T_{\mu_1}\widetilde{T}_{\mu_2}.
\end{equation}
It can be shown that (\ref{eq:ansatz1}) implies $V$ is a rank-$1$ tensor operator (see \cite{barut}). Substituting the coefficients from appendix \ref{app:CGtable} one gets, in terms of the generators,
\begin{equation}
J_\pm=\pm i T_\pm \widetilde{T}_\pm,\quad J_0=-\tfrac{1}{2} \left( T_- \widetilde{T}_+ + T_+ \widetilde{T}_- \right),
\end{equation}
with the shorthand notation
\begin{equation}
T_\pm:=T_{\pm\frac{1}{2}}.
\end{equation}
Since the vector operator $V$ is constrained to map each representation to itself by the generators, the additional assumption
\begin{equation}\label{eq:ansatz2}
\braket{j'\Vert T \Vert j}=f(j)\,\delta_{j',j+\frac{1}{2}},
\quad
\braket{j'\Vert \widetilde{T} \Vert j}=\widetilde{f}(j)\,\delta_{j',j-\frac{1}{2}}
\end{equation}
will be made, with $f$ and $\widetilde f$ arbitrary functions.

The matrix elements of the generators are known. Using the ansatz, one gets
\begin{equation}
C_+(j,m)=\braket{j,m+1|J_+|j,m}=i\braket{j,m+1|T_+|j-\tfrac{1}{2},m+\tfrac{1}{2}}\braket{j-\tfrac{1}{2},m+\tfrac{1}{2}|\widetilde{T}_+|j,m}.
\end{equation}
The right hand side can be evaluated using the Wigner--Eckart theorem, assuming the decomposition $F_{\tfrac{1}{2}}\otimes \rho_j$
exists. One gets that the r.h.s is  
\begin{equation}
i\frac{f(j-\tfrac{1}{2})\widetilde{f}(j)}{\sqrt{2j}\sqrt{2j+1}}\sqrt{j-m}\sqrt{j+m+1}=\frac{f(j-\tfrac{1}{2})\widetilde{f}(j)}{\sqrt{2j}\sqrt{2j+1}}\, C_+(j,m),
\end{equation}
so that it must be
\begin{equation}\label{eq:ansatz_condition}
\frac{f(j-\tfrac{1}{2})\widetilde{f}(j)}{\sqrt{2j}\sqrt{2j+1}} = 1.
\end{equation}
The same constraint is obtained by repeating the argument for $J_-$ and $J_0$, which means the ansatz is true whenever (\ref{eq:ansatz_condition}) holds. To simplify notation the choice
\begin{equation}
f(j)=\widetilde{f}(j)=\sqrt{2j+1}
\end{equation}
will be used here.

The action of $T$ and $\widetilde T$ is found to be
\begin{subequations}
\begin{align}
T_-\ket{j,m}&=\sqrt{j-m+1}\ket{j+\tfrac{1}{2},m-\tfrac{1}{2}}\\
T_+\ket{j,m}&=\sqrt{j+m+1}\ket{j+\tfrac{1}{2},m+\tfrac{1}{2}}\\
\widetilde{T}_-\ket{j,m}&=-\sqrt{j+m}\ket{j-\tfrac{1}{2},m-\tfrac{1}{2}}\\
\widetilde{T}_+\ket{j,m}&=\sqrt{j-m}\ket{j-\tfrac{1}{2},m+\tfrac{1}{2}},
\end{align}
\end{subequations}
from which it follows that
\begin{equation}\label{eq:heisenberg_cr}
[T_+,\widetilde{T}_-]=[\widetilde{T}_+,T_-]=\mathbbm{1},
\end{equation}
with all other commutators vanishing.

These commutation relations closely resemble those of the harmonic oscillator and, in fact generalize them. For example, when the representation considered is $F_j$, with $j\geq\frac{1}{2}$, one finds by inspection
\begin{equation}
T_\pm=\mp \widetilde{T}^\dagger_\mp.
\end{equation}
Renaming
\begin{equation}
\widetilde{T}_-=a,\quad \widetilde{T}_+=-b
\end{equation}
one finds
\begin{equation}
J_+=ia^\dagger b,\quad J_-=i b^\dagger a,\quad J_0=\tfrac{1}{2}\left(a^\dagger a - b^\dagger b \right),
\end{equation}
with $a$ and $b$ satisfying the harmonic oscillator commutation relation.

Analogously, for the discrete series $D^\pm_j$, with $j\geq 0$, one has
\begin{equation}
T_\pm=
\begin{cases}
-\widetilde{T}^\dagger_\mp &\quad \mbox{for }D^+_j\\
\widetilde{T}^\dagger_\mp &\quad \mbox{for }D^-_j.
\end{cases}
\end{equation}
With the choice
\begin{equation}
\begin{cases}
\widetilde{T}_-=a,\quad \widetilde{T}_+=ib^\dagger &\quad \mbox{for }D^+_j\\
\widetilde{T}_-=a^\dagger,\quad \widetilde{T}_+=ib &\quad \mbox{for }D^-_j
\end{cases}
\end{equation}
one gets
\begin{equation}
J_+=a^\dagger b^\dagger,\quad J_-=a b,\quad J_0=
\begin{cases}
\tfrac{1}{2}\left(a^\dagger a + b^\dagger b +1\right) &\quad \mbox{for }D^+_j\\
-\tfrac{1}{2}\left(a^\dagger a + b^\dagger b +1\right) &\quad \mbox{for }D^-_j.
\end{cases}
\end{equation}

The continuous series generators cannot be rewritten in terms of harmonic oscillators because, while
\begin{equation}
\braket{j+\tfrac{1}{2},m\pm \tfrac{1}{2}|T_\pm|j,m}=\braket{j,m|\widetilde{T}_\mp|j+\tfrac{1}{2},m\pm \tfrac{1}{2}},
\end{equation}
these matrix elements are never always real or imaginary, as that depends on the value of $m$. This is to be expected, as if the generators could be written in terms of harmonic oscillators, the Casimir element $Q$ would be expressible in terms of the \emph{number operators}
\begin{equation}
N_a=a^\dagger a,\quad N_b=b^\dagger b,
\end{equation}
which have discrete spectrum \cite{messiah1}: this contradicts the fact that the eigenvalues of $Q$ are continuous. 
Nevertheless, an analogue of the Jordan--Schwinger representation exists in this case. One should note that the commutation relations (\ref{eq:heisenberg_cr}) are still those of a Heisenberg algebra representation, where one of the generators acts as the identity. 
\section*{Concluding remarks}

The methods used to study the recoupling theory of finite and infinite-dimensional representations heavily relies on the particular Lie algebra being considered. Nevertheless, they can hopefully serve as a guide when considering more general algebras, e.g. $\mathfrak{sl}(2,\mathbb{C})_\mathbb{R}$; it is however likely that the same techniques will work with little modification with the deformed algebra $\mathcal{U}_q(\sltr)$. Both directions could be investigated in the future.

Regarding the Jordan-Schwinger representation, one question arises: the tensor operators used to construct the generators satisfy the commutation relations of a Heisenberg algebra, i.e. they form a representation of it. What kind of representation it is, and is it unitary or not? This aspect has not been considered in detail yet. The author leaves it for further investigations.

\begin{acknowledgments}
The author would like to thank Florian Girelli for introducing him to the topic and providing useful insights.
\end{acknowledgments}

\appendix
\section{Notations}\label{app:notations}
\begin{tabular}{lll}
\multicolumn{3}{l}{Braket notation}\\ \midrule
$\ket{\psi}$ & &	vector in an Hilbert space\\
$\braket{\varphi|\psi}$ & & inner product of $\ket\varphi$ and $\ket\psi$, antilinear in $\ket\varphi$\\
$A^\dagger$ && Hermitian adjoint of an operator $A$\\[1em]
\multicolumn{3}{l}{Representation theory}\\ \midrule
$\mathfrak{gl}(V)$	& & algebra of endomorphisms of a vector space $V$ \\
$\mathrm{Lin}(V_1,V_2)$	& &	vector space of linear maps $V_1\rightarrow V_2$\\
$V\otimes W$	& &	(orthogonal) tensor product of two vector spaces $V$, $W$\\
$V\oplus W$	& &	(orthogonal) direct sum of two vector spaces $V$, $W$\\[1em]
\multicolumn{3}{l}{Sets}\\ \midrule
$x+\mathbb{Z}$ & & set defined by $\left\{ x+z\in \mathbb{R} \left|\right. z\in\mathbb{Z} \right\}$, assuming $x\in \mathbb{R}$\\
$y\mathbb{Z}$ & & set defined by $\left\{ yz\in \mathbb{R} \left|\right. z\in\mathbb{Z} \right\}$, assuming $y\in \mathbb{R}$
\end{tabular}

\section{Some properties of the Clebsch--Gordan coefficients}
\label{app:CGprop}

Some properties of the Clebsch--Gordan coefficients will be listed here. Assume that the coupling $F_\gamma\otimes \rho_j$, with $\rho_j$ an arbitrary irreducible representation, is decomposable. Consider the Clebsch--Gordan coefficients in the form presented in section \ref{sec:WE}, that is such that the diagonalized basis vectors are
\begin{equation}
\label{eq:CG_recur0}
\ket{J,M}=\sum_{\mu\in\mathcal{M}_\gamma}\sum_{m\in\mathcal{M}_j}A(\gamma,\mu;j,m|J,M)\ket{\gamma,\mu;j,m},\quad J\in\mathcal{D}(\gamma,j),\quad M\in\mathcal{M}_J.
\end{equation}
One can always rescale these vectors so that
\begin{equation}
\mathcal{J_\pm}\ket{J,M}=C_\pm(J,M)\ket{J,M\pm 1}.
\end{equation}
By acting with $\mathcal{J_\pm}$ on both sides of (\ref{eq:CG_recur0}) and equating the coefficients of each basis vector one obtains that the Clebsch--Gordan coefficients must obey the recursion relation
\begin{equation}\label{eq:CG_recur1}
C_\pm(J,M)A(\gamma,\mu;j,m|J,M\pm 1)=
C_\pm(\gamma,\mu \mp 1)A(\gamma,\mu \mp 1;j,m|J,M) +
C_\pm(j,m\mp 1)A(\gamma,\mu;j,m\mp 1|J,M);
\end{equation}
analogously, one finds for the inverse coefficients
\begin{equation}\label{eq:CG_recur2}
C_\pm(J,M)B(J,M\pm 1|\gamma,\mu;j,m)=
C_\pm(\gamma,\mu \mp 1)B(J,M|\gamma,\mu \mp 1;j,m) +
C_\pm(j,m\mp 1)B(J,M|\gamma,\mu;j,m\mp 1).
\end{equation}
Since both the coefficients and their inverse, for each fixed $J$, are solutions the same homogeneous linear system, they must be proportional to each other: one can always choose their normalization so that
\begin{equation}
A(\gamma,\mu;j,m|J,M)=B(J,M|\gamma,\mu;j,m).
\end{equation}
Since the recursion relations only relate coefficients with the same $J$, one could think \emph{a priori} that coefficients with different $J$ are independent. It will be shown in the following that this is not true.

Consider the particular case of (\ref{eq:CG_recur1})
\begin{equation}
C_+(J,M)B(J,M+1|\gamma,-\gamma;j,m+1)=C_+(j,m)B(J,M|\gamma,-\gamma;j,m),
\end{equation}
where the fact that
\begin{equation}
C_+(\gamma,-\gamma-1)=0
\end{equation}
was used. By considering the same equation for $J+1$ and dividing by the first one, one obtains
\begin{equation}
D_J(M+1):=\frac{B(J+1,M+1|\gamma,-\gamma;j,m+1)}{B(J,M+1|\gamma,-\gamma;j,m+1)}=\frac{C_+(J,M)}{C_+(J+1,M)}\frac{B(J+1,M|\gamma,-\gamma;j,m)}{B(J,M|\gamma,-\gamma;j,m)},
\end{equation}
i.e.
\begin{equation}
D_J(M+1)=\frac{\sqrt{J-M}\sqrt{J+M+1}}{\sqrt{J-M+1}\sqrt{J+M+2}}D_J(M).
\end{equation}
One can readily see that, by recursion,
\begin{equation}
D_J(M+n)=\frac{\sqrt{J-M-n+1}\sqrt{J+M+1}}{\sqrt{J+M+n+1}\sqrt{J-M+1}}D_J(M),\quad n\in\mathbb{N}.
\end{equation}
From this one can infer that
\begin{equation}
D_J(M)\equiv \frac{B(J+1,M|\gamma,-\gamma;j,m)}{B(J,M|\gamma,-\gamma;j,m)}=\alpha(J)\frac{\sqrt{J-M+1}}{\sqrt{J+M+1}},
\end{equation}
where $\alpha$ is arbitrary and depends on the normalization. 
\section{Table of Clebsch--Gordan coefficients for \texorpdfstring{$\gamma=\frac{1}{2}$ and $\gamma=1$}{y=1/2 and y=1}}
\label{app:CGtable}
Explicit values for the Clebsch--Gordan coefficients are presented here, for the small values $\gamma=\frac{1}{2}$ (table \ref{tab:1/2}) and $\gamma=1$ (table \ref{tab:1}), with arbitrary $j$. The tables are valid for $D^\pm_j$, $C^\varepsilon_j$ and $F_j$, provided only the allowed values of $j$, $J$ and $M$ are considered (see \ref{sec:fin-dis} and \ref{sec:fin-cont}). The coefficients are normalized in such a way that
\begin{equation}
A(\gamma,\mu;j,m|J,M)=B(J,M|\gamma,\mu;j,m)
\end{equation}
and that, for the finite-dimensional series (with $j\geq \gamma$), they coincide with the $\mathfrak{su}(2)$ ones. Moreover, in analogy with the $\mathfrak{su}(2)$ case, the Clebsch--Gordan coefficients for the coupling $\rho_j\otimes F_\gamma\cong F_\gamma\otimes \rho_j$ are chosen to be
\begin{equation}
B(J,M|j,m;\gamma,\mu):=(-1)^{J-j-\gamma}B(J,M|\gamma,\mu;j,m).
\end{equation}
\begin{center}
\begin{minipage}[t]{0.8\textwidth}
\[
\begin{array}{lcc} \toprule
& J=j-\frac{1}{2} &  J=j+\frac{1}{2}  \\ \midrule

\mu=-\frac{1}{2}\quad &-\frac{\sqrt{j+M+\frac{1}{2}}}{\sqrt{2j+1}} & \frac{\sqrt{j-M+\frac{1}{2}}}{\sqrt{2j+1}} \\[0.8em]
\mu=+\frac{1}{2} & \frac{\sqrt{j-M+\frac{1}{2}}}{\sqrt{2j+1}} & \frac{\sqrt{j+M+\frac{1}{2}}}{\sqrt{2j+1}} \\ \bottomrule
\end{array}
\]
\captionof{table}{\label{tab:1/2}Clebsch--Gordan coefficients $B(J,M|\gamma,\mu;j,M-\mu)$ for $\gamma=\frac{1}{2}$.}
\[
\begin{array}{lccc} \toprule
&\quad J=j-1\quad &\quad  J=j\quad & \quad J=j+1 \quad\\ \midrule

\mu=-1\quad &\frac{\sqrt{j+M}\sqrt{j+M+1}}{\sqrt{2j}\sqrt{2j+1}}  &
-\sqrt{2}\frac{\sqrt{j-M}\sqrt{j+M+1}}{\sqrt{2j}\sqrt{2j+2}} & 
\frac{\sqrt{j-M}\sqrt{j-M+1}}{\sqrt{2j+1}\sqrt{2j+2}} \\[0.8em]

\mu=0 &-\sqrt{2}\frac{\sqrt{j-M}\sqrt{j+M}}{\sqrt{2j}\sqrt{2j+1}} & -\frac{2M}{\sqrt{2j}\sqrt{2j+2}} &
\sqrt{2}\frac{\sqrt{j-M+1}\sqrt{j+M+1}}{\sqrt{2j+1}\sqrt{2j+2}}\\[0.8em]

\mu=+1 & \frac{\sqrt{j-M}\sqrt{j-M+1}}{\sqrt{2j}\sqrt{2j+1}} & 
\sqrt{2} \frac{\sqrt{j+M}\sqrt{j-M+1}}{\sqrt{2j}\sqrt{2j+2}}&
\frac{\sqrt{j+M}\sqrt{j+M+1}}{\sqrt{2j+1}\sqrt{2j+2}} \\ \bottomrule
\end{array}
\]
\captionof{table}{\label{tab:1}Clebsch--Gordan coefficients $B(J,M|\gamma,\mu;j,M-\mu)$ for $\gamma=1$.}
\end{minipage}
\end{center}

\section{Tridiagonal matrices}\label{app:tridiagonal}
A \emph{tridiagonal} matrix is a square matrix whose only non-zero entries are on the diagonal, the \emph{subdiagonal}~---~consisting of the entries directly below and to the left of the diagonal~---~and the \emph{superdiagonal}~---~consisting of the entries directly above and to the right of the diagonal. They can be visualized as 
\begin{equation}
A=
\begin{pmatrix}
b_1 & c_1\\
a_2 & b_2 & c_2\\
& \ddots & \ddots & \ddots\\
&& a_{n-1} &b_{n-1} & c_{n-1}\\
& & & a_n & b_n
\end{pmatrix},
\end{equation}
with the generic entry given by
\begin{equation}
\label{eq:tridiagonal}
A_{ij}=\delta_{i-1,j}\, a_i + \delta_{i,j}\,b_i + \delta_{i+1,j}\, c_i,
\end{equation}
where it is understood that
\begin{equation}
a_1=0 \quad \mbox{and}\quad c_n=0.
\end{equation}
A result holding for a certain class of tridiagonal matrices will be proved here.
\begin{appxprop}
Let $A$ be a $n\times n$ tridiagonal matrix over a field $\mathbb{K}$, with all non-zero superdiagonal (subdiagonal) entries; then, for any $\lambda\in\mathbb{K}$, the kernel of
\begin{equation}
A-\lambda\mathbbm{1}
\end{equation}
is at most 1-dimensional.
\end{appxprop}
\begin{proof}
First consider the superdiagonal case. Fix an arbitrary $\lambda\in\mathbb{K}$; the kernel of~$A-\lambda\mathbbm{1}$~is the vector space of solutions to the equation
\begin{equation}
Ax=\lambda x,\qquad x\in \mathbb{K}^n,
\end{equation}
which, with the notation introduced in eq. (\ref{eq:tridiagonal}), is equivalent to the system of equations
\begin{equation}
\begin{cases}
\left(b_1-\lambda\right)x_1 + c_1\,x_2 = 0\\
a_i\, x_{i-1} + \left(b_i-\lambda\right)x_i + c_i\,x_{i+1}=0,\qquad i=2,\dotsc,n-1\\
a_n\,x_{n-1} + \left(b_n-\lambda\right)x_n=0.
\end{cases}
\end{equation}
If $x_1$ is zero, the first equation reduces to
\begin{equation}
c_1\,x_2=0,
\end{equation}
which, since all $c$'s are non-vanishing, implies $x_2$ is zero as well. In general, the $k$th equation will be
\begin{equation}
c_k\,x_{k+1}=0,
\end{equation}
so that $x$ must necessarily be the zero vector.

Now let $x_1$ be an arbitrary non-zero value. By plugging each equation in the next one, the first $n-1$ equations will give a system of equations of the form
\begin{equation}
\label{eq:tridiagonal_solution}
c_i\,x_{i+1}=\alpha_{i+1}\,x_1,\qquad i=1,\dotsc,n-1,
\end{equation}
with each $\alpha$ depending solely on $\lambda$ and on the matrix entries.
These always have solution, since one can safely divide by the $c$'s; as a consequence, the solution is \emph{completely} specified by the value of $x_1$, which can be factored out as a scalar coefficient. The $n$th equation, together with the $(n-2)$th one, will give
\begin{equation}
\left(b_n-\lambda\right)x_n=-\frac{a_n\,\alpha_{n-1}}{c_{n-2}}\,x_1;
\end{equation}
a non-zero eigenvector exists if and only if the solution to this equation is consistent with the others. By virtue of equations (\ref{eq:tridiagonal_solution}), all such solutions are proportional to each other, thus
\begin{equation}
\mathrm{dim}\ker\left(A-\lambda\mathbbm{1}\right)\leq 1.
\end{equation}
As for the subdiagonal case, it can be reduced to the superdiagonal one by working with the transpose matrix $A^\transpose$. From the \emph{fundamental theorem of linear algebra} follows that, for square matrices,
\begin{equation}
\ker\left(A^\transpose-\lambda\mathbbm{1}\right)\equiv\mathrm{coker}\left(A-\lambda\mathbbm{1}\right)\cong \ker\left(A-\lambda\mathbbm{1}\right),
\end{equation}
so that, again, the kernel is at most $1$-dimensional.
\end{proof}
It trivially follows that

\begin{appxcoroll}\label{coroll:tridiagonal_eigenspaces}
The eigenspaces of a tridiagonal matrix over a field $\mathbb{K}$, whose superdiagonal (subdiagonal) entries are all non-zero, are all $1$-dimensional.
\end{appxcoroll}

\bibliography{bibliography}
\end{document}